\documentclass[11pt]{article}

\usepackage{epsfig,amsmath,latexsym,amssymb}
\usepackage{graphicx}
\usepackage{lscape}
\usepackage{picture, eso-pic, tikz} 
\usepackage{dsfont}
\usepackage{listings}
\usepackage{changes}
\usepackage{bbding}
\usepackage{url}
\usepackage{hyperref}
\usepackage{tablefootnote}

\usepackage{tikz}
\usetikzlibrary{arrows,shapes,er,positioning,arrows.meta,calc}
\tikzset{
     arrow/.style = { thick,  ->, >=Triangle},
}

\renewcommand{\baselinestretch}{1.1}
\def\E{{\mathbb E}}  %
\def\Beweis{\footnotesize}

\newcommand{\Remm}[1]{}
\newtheorem{theo}{Theorem}[section]
\newtheorem{lemma}[theo]{Lemma}

\newtheorem{cor}[theo]{Corollary}
\newtheorem{defi}[theo]{Definition}
\newtheorem{model ass}[theo]{Model Assumptions}
\newtheorem{estimator}[theo]{Estimator}

\newtheorem{rems}[theo]{Remarks}

\def\EndProof{{\begin{flushright}\vspace{-2mm}$\Box$\end{flushright}}}

\numberwithin{equation}{section}

\definecolor{MyGray}{rgb}{0.92,0.92,0.92}


\definecolor{British racing}{rgb}{0.0, 0.5, 0.0}

\begin{document}
\author{Ronald Richman\footnote{insureAI, ronaldrichman@gmail.com}
  \and
  Mario V.~W\"uthrich\footnote{Department of Mathematics, ETH Zurich,
mario.wuethrich@math.ethz.ch}}

\date{Version of \today}
\title{A Note on the Generalized Cape Cod Reserving Method}
\maketitle

\begin{abstract}
\noindent Claims reserving is one of the most important actuarial tasks in non-life insurance modeling. There are several popular methods to perform claims reserving such as the chain-ladder (CL), the Bornhuetter--Ferguson (BF) or the generalized Cape Cod (GCC) methods. These methods have originally been introduced as deterministic algorithms, and only in a later step, they have been lifted to stochastic models allowing for analyzing claims prediction uncertainty. This holds true for the CL and the BF methods, but not for the GCC method. The purpose of this article is to close this gap and derive an analytical formula for the mean squared error of prediction (MSEP) of the GCC method.

\medskip

\noindent
{\bf Keywords.} Claims reserving, chain-ladder, Bornhuetter--Ferguson, Cape Cod, generalized Cape Cod, prediction uncertainty, mean squared error of prediction, MSEP. 

\end{abstract}

\section{Introduction}

Claims reserving is one of the core forecasting problems in non-life insurance operations. Often, the claims reserves constitute the largest item on the liability side of the balance sheet of a non-life insurance company. Therefore, the reserving problem is not only about a good point prediction of ultimate losses, but also about understanding the underlying uncertainties in this point prediction. Moreover, the results produced during the reserving exercise are used to understand the most recent loss experience as it emerges. This provides important information to guide management of insurers.
In the insurance industry, the chain-ladder (CL) method \cite{Mack}, the Bornhuetter--Ferguson (BF) method \cite{BF}, and the generalized Cape Cod method \cite{Bu_unpublished, Gluck} remain the most popular techniques to set the claims reserves.

Over the past decades, the stochastic claims reserving literature has given different probabilistic interpretations of these classical claims reserving algorithms. For the CL method, Mack’s \cite{Mack} distribution-free CL model has been introduced for estimating the prediction uncertainty in the CL reserve estimates. Renshaw--Verrall \cite{RV} then showed that the CL algorithm can also be embedded into a generalized linear model (GLM) using quasi-likelihood estimation, while England--Verrall \cite{EV} broadened the picture by surveying various stochastic reserving models that reproduce classical estimates and by developing bootstrap and Bayesian approaches to full predictive distributions. This stream of literature shows that a deterministic reserving algorithm can have many different underlying stochastic models producing the same claims reserves, but providing different estimations of the prediction uncertainty in these claims reserves.

A parallel line of work has developed BF-type methods and their credibility interpretations. The original BF idea \cite{BF} combines observed loss development with a prior estimate of ultimate losses, making it especially attractive when observed claim development is immature or unstable for a pure observation-based extrapolation. This perspective was strengthened by Bayesian and empirical Bayes reformulations: Verrall \cite{Verrall} connected CL and BF-style reserving to Bayesian and generalized linear modeling, Mack \cite{Mack2000} showed that the Benktander \cite{Benktander} and Hovinen \cite{Hovinen} method can be interpreted as an intuitive credibility blend of the CL and BF methods, and Schmidt--Zocher \cite{SchmidtZocher} proposed the BF principle as a unifying framework for a family of reserve predictors rather than a single algorithm. The BF method received dedicated stochastic treatments from Mack \cite{Mack2008}, Alai et al.~\cite{AMW1, AMW2}, Saluz et al.~\cite{Saluz2011}, and a Bayesian formulation by England et al.~\cite{EVW}. These latter frameworks allow for the estimation of prediction uncertainty under the different stochastic formulations of the BF method.
An important (and critical) step to derive a stochastic model for the BF is to model the selection of the expected loss ratio (prior estimate) in an appropriate manner. This is challenging since this prior estimate should be set ``independently" of the observed data.

Against this background, the treatment of the uncertainty evaluation within the Cape Cod (CC) method of B\"uhlmann \cite{Bu_unpublished} and the generalized Cape Cod (GCC) method of Gluck \cite{Gluck} remains sparse. The CC method was introduced in 1983 in a summer school taught by B\"uhlmann \cite{Bu_unpublished}, and this method has not been published in any peer-reviewed actuarial journal. It presents an interesting middle ground between prior-based reserving (BF) and fully observation-based reserving (CL), by integrating an {\it estimated} loss ratio that is calculated from prior information, typically earned premium, and claims experience. This is precisely the feature that makes the CC method attractive in practice when one wants more stability than the CL method can provide for immature years, but it has also the feature that it is more responsive to observations than the BF method, because it does not fully rely on a prior expected loss ratio which may not be fully reliable. Gluck’s GCC \cite{Gluck} proposal allows for an additional trending by considering a chronologically ordered weighted average for the estimation of the loss ratio. However, the literature on the CC and the GCC methods remained largely algorithmic and our goal is to close this gap.

Precisely because the CC and the GCC methods utilize the observed information to estimate the loss ratio and does not require an ``independent'' loss ratio estimate, it should, in theory, be a promising candidate for stochastic modeling of reserves beyond the CL method. For analyzing prediction uncertainty within the CC method, Saluz \cite{Saluz} made an important first step towards stochastic modeling by introducing a distribution-free stochastic model that allowed her to estimate the prediction uncertainty within the CC method. However, her estimate does not fully match the application of the CC method in practice, because she used parameter estimates for CC reserving that are different from the ones typically used in industry -- the industry standard uses CL-implied parameters, whereas Saluz \cite{Saluz} used likelihood-implied parameters. This paper provides a different approach to the CC uncertainty estimation that is fully in line with the industry standard. We derive these results within the broader framework of GCC reserving, which encompasses the CL and the CC methods as the two boundary cases. Consequently, the results for these two limiting cases follow directly, as a natural byproduct of our derivations. The main tool to derive the prediction uncertainty estimates within the GCC framework is the error propagation method of R\"ohr \cite{Rohr}. R\"ohr \cite{Rohr} demonstrated that using error propagation considerations, one can derive Mack's \cite{Mack} uncertainty estimate along a different route. This paper demonstrates how the error propagation methodology can be used within the GCC framework for uncertainty estimation. One resulting boundary case exactly provides us with Mack's \cite{Mack} and R\"ohr's \cite{Rohr} estimates, and the other boundary case gives us a new uncertainty estimate for the CC method of B\"uhlmann \cite{Bu_unpublished}.

The practical relevance of this work is to provide a mean squared error of prediction (MSEP) estimator that mirrors the predictor that practitioners actually book. Many industry reserving processes use the BF method, but none of the proposed BF MSEP estimators has achieved widespread practical use. A plausible explanation is that the existing BF MSEP estimators do not closely mirror the way the BF method is applied in practice -- most importantly, they often require an explicit probabilistic specification of the prior expected loss ratio, whereas in industry the prior is typically obtained from a planning or pricing process whose stochastic properties are not formalized. The GCC method provides a good algorithmic approximation to the BF method, with the prior loss ratio replaced by a chronologically-weighted average of CL-implied loss ratios. The framework presented here delivers a closed-form MSEP estimator for the GCC method, fully consistent with the CL-implied parameter estimates that are the industry standard. In this sense, the present work supplies an analytical uncertainty estimator for what is, in effect, a good algorithmic proxy for the BF method, that is in  itself in widespread practical use.

In a similar manner, a second point of practical relevance concerns the Mack bootstrap of England--Verrall \cite{EnglandVerrall2006}, which is widely used in industry to estimate predictive distributions for reserves. The Mack bootstrap reproduces the CL reserves in the mean, which is typically different from the reserves actually booked when the booked reserves are obtained by a CC, GCC, or BF procedure. Often, some scaling is done to try to close the gap between the mean reserves produced by the Mack bootstrap procedure and what has actually been booked. The MSEP estimator presented here closes this gap with a closed-form solution: the same CL-implied loss ratios that drive the booked GCC reserve drive also the corresponding RMSEP figure.

\medskip

{\bf Organization.}
The next section introduces the CC predictor, it discusses parameter estimation in the CC method, and it connects the CC to the CL method. In the next step, we introduce the GCC method and we give several credibility interpretations of the GCC method. Section \ref{sec: Prediction uncertainty} derives the prediction uncertainty formulas measured by the MSEP. We start with Mack's \cite{Mack} CL MSEP formula and we explain how it can be derived using the error propagation technique. This sets the ground to derive an MSEP formula for the GCC using error propagation in Section \ref{sec: Error propagation in the Cape Cod method}.
Section \ref{sec: Example} gives a numerical example which is compared to the CC results of Saluz \cite{Saluz}. Finally, in Section \ref{sec: Conclusions} we conclude.

\section{Cape Cod methods}
\label{sec: Cape Cod methods}
\subsection{B\"uhlmann's Cape Cod method}
We start by discussing the CC reserving method introduced by B\"uhlmann \cite{Bu_unpublished}. The following model assumptions give a very general set-up for the CC reserving method, and we are going to consider more confined assumptions later on.

Let $C_{i,j}$ denote the cumulative claims in accident year $i \in \{1, \ldots, I\}$ after development period $j \in \{0,\ldots, J\}$. The meaning of cumulative claims if fairly general, e.g., $C_{i,j}$ can reflect cumulative payments, claims incurred or total number of reported claims. The total claim of accident year $i \in \{1, \ldots, I\}$ is given by the ultimate claim $C_{i,J}$. This ultimate claim is the main object of interest that we aim to forecast.

\begin{model ass}~\label{Model-Assumptions-Cape-Cod}
~
\begin{itemize}
\item Cumulative claims $(C_{i,j})_{0\le j \le J}$ of different accident
years $1\le i \le I$ are independent.
\item There are positive premiums $\pi_i>0$ for all accident years $1\le i \le I$, a claims ratio $\kappa>0$ and a positive claims development pattern $(\beta_j)_{0\leq j\leq J}$ with normalization
$\beta_J=1$ such that
\begin{equation*}
\E [C_{i,j}]= \beta_j\,\kappa \, \pi_i, 
\end{equation*}
for all $1 \le i \le I$ and $0\le j \le J$.
\end{itemize}
\end{model ass}

For simplicity, we assume $I=J+1$, thus, we work on a claims development triangle; for $I>J+1$ we would have a trapezoid, which makes the indexing unnecessarily tedious.

\begin{rems}\normalfont
\begin{itemize}
\item We call $\pi_i>0$ the premium of accident year $i$. This terminology is useful for getting an intuitive interpretation of the claims ratio $\kappa$. More generally, $(\pi_i)_{i=1}^I$ are just exposures for the corresponding accident years. Practically, usually the earned premium of calendar year $i$ is used for the exposure $\pi_i$.
\item We can view Model Assumptions \ref{Model-Assumptions-Cape-Cod} as a minimal set of assumptions. B\"uhlmann \cite{Bu_unpublished} defines an explicit frequency-severity model which complies with the above model assumptions. For explicit mathematical results, one often decomposes the cumulative claims into increments $Y_{i,0}, \ldots, Y_{i,J}$ such that $C_{i,j}=\sum_{l=0}^j X_{i,l}$. These increments have means 
\begin{equation}\label{cross-classified model}
\E [Y_{i,j}]= \gamma_j\,\kappa \, \pi_i, 
\end{equation}
with the incremental claims pattern $\gamma_j=\beta_j - \beta_{j-1}$, for $1\le j \le J$, and $\gamma_0=\beta_0$. 
\item Under an additional independence assumption on the increments $Y_{i,0}, \ldots, Y_{i,J}$, \eqref{cross-classified model} defines a classical cross-classified claims reserving model, e.g., the cross-classified over-dispersed Poisson (ODP) model of England--Verrall \cite{EV} satisfies \eqref{cross-classified model}. Moreover, the CC model in Saluz \cite[Model Assumptions 2.1]{Saluz} additionally makes a specific variance assumption. For the moment, we want to consider our more general Model Assumptions \ref{Model-Assumptions-Cape-Cod}, in particular, we do not want to assume that incremental claims are independent. 
\end{itemize}
\end{rems}

At time $I$, we have observed the cumulative claims $C_{i,I-i}$, $1 \le i \le I$. This gives us an unbiased estimator for the claims ratio $\kappa$
\begin{equation}\label{individual claims ratios}
\widehat{\kappa}_i=\frac{C_{i,I-i}}{\beta_{I-i}\, \pi_i}
\qquad \text{ for $1 \le i \le I$.}
\end{equation}
Aggregating this over all observed accident years gives us the
CC estimator for the claims ratio $\kappa$
\begin{equation}\label{CC kappa}
\widehat{\kappa}^{\rm CC}
=\frac{\sum_{i=1}^I C_{i,I-i}}
{\sum_{i=1}^I \beta_{I-i}\, \pi_i}
=\sum_{i=1}^I \frac{\beta_{I-i}\, \pi_i}{\sum_{k=1}^I
\beta_{I-k}\, \pi_k}\, \widehat{\kappa}_i.
\end{equation}
\begin{lemma} 
The CC estimators $\widehat{\kappa}_i$ and $\widehat{\kappa}^{\rm CC}$ are unbiased for $\kappa$.
\end{lemma}
The proof of this lemma is straightforward.
Under an additional variance assumption 
\begin{equation}\label{BLUE assumption}
  {\rm Var} (C_{i,j})= \beta_j\, \tau^2 \, \pi_i
  \qquad \text{ for $1\le i \le I$ and $0\le j \le J$,}
\end{equation}
the estimator $\widehat{\kappa}^{\rm CC}$ for $\kappa$ is the best linear unbiased estimator (BLUE) for the claims ratio $\kappa$; in our further developments we are not going to rely on such a variance assumption \eqref{BLUE assumption}.

\begin{defi}[CC predictor]
The CC predictors at time $I$ for the ultimate claims of accident years $I-J+1\le i \le I$ are defined by
\begin{equation}
\widehat{C}_{i,J}^{\rm CC}= C_{i,I-i} + \left(1-\beta_{I-i}\right) 
\widehat{\kappa}^{\rm CC} \,\pi_i.
\end{equation}
\end{defi}
We generally assume that the premiums $\pi_i>0$ are given and known, and the general interest lies in estimating the claims development pattern $(\beta_j)_{0\leq j\leq J}$ for explicitly computing the CC predictors
$\widehat{C}_{i,J}^{\rm CC}$. Before discussing this in detail, we would like to relate this claims development pattern $(\beta_j)_{0\leq j\leq J}$ to the chain-ladder (CC) factors of Mack's CL model \cite{Mack}.

\begin{model ass}[Mack's distribution-free CL model \cite{Mack}]
\label{distribution-free CL model assumptions}~
\begin{itemize}
\item Cumulative claims $(C_{i,j})_{0\le j \le J}$ of different accident
years $1\le i \le I$ are independent and strictly positive, a.s.
\item There exist positive CL factors $(f_j)_{j=0}^{J-1}$ and positive variance parameters  $(\sigma^2_j)_{j=0}^{J-1}$
such that for all $1\le i \le I$ and $0 \le j  \le J-1$
\begin{eqnarray*}
\E \left[\left.C_{i,j+1}\right| C_{i,0}, \ldots, C_{i,j}
\right]&=& f_j\, C_{i,j},\\
{\rm Var} \left(\left.C_{i,j+1}\right| C_{i,0}, \ldots, C_{i,j}
\right)&=& \sigma^2_j\, C_{i,j}.
\end{eqnarray*}
\end{itemize}
\end{model ass}
An immediate consequence of Mack's CL model assumptions is
\begin{equation*}
\E \left[C_{i,j}\right] = \prod_{k=0}^{j-1}f_k\, \E \left[C_{i,0}\right]  
= \frac{\prod_{k=0}^{J-1}f_k}{\prod_{k=j}^{J-1}f_k}\, \E \left[C_{i,0}\right]
= \beta_j\, \frac{\E \left[C_{i,0}\right]}{\beta_0},
\end{equation*}
where we identify the parameters (an empty product is set equal to 1) as follows
\begin{equation}\label{pattern known CL}
\beta_{j}=  \prod_{k=j}^{J-1}f^{-1}_{k} \qquad \text{ for $0 \le j\le J$}.
\end{equation}
By setting $\kappa \pi_i = \E[C_{i,0}]/\beta_0$, we note that Mack's CL model satisfies the CC Model Assumptions \ref{Model-Assumptions-Cape-Cod}. 
Naturally, a constant claims ratio $\kappa$ in the identity $\kappa \pi_i = \E[C_{i,0}]/\beta_0$ means that we assume stationarity in the claims ratios and that the exposures $\pi_i$ properly discriminate the expected ultimate claims $\E[C_{i,J}]$ of the different accident years $i$.
This connection then provides another interesting view on the CC method; see Saluz \cite[Section 1.3]{Saluz}. Namely, the CC predictor is given by
\begin{equation}\label{CC predictor 00}
  \widehat{C}_{i,J}^{\rm CC}= C_{i,I-i} + \beta_{I-i}\,\widehat{\kappa}^{\rm CC} \,\pi_i
  \left(\prod_{j=I-i}^{J-1}f_{j}-1\right),
\end{equation}
whereas the corresponding CL predictor is given by
\begin{equation}\label{CL predictor 00}
  \widehat{C}_{i,J}^{\rm CL}= C_{i,I-i} \prod_{j=I-i}^{J-1}f_{j}
=  C_{i,I-i} + C_{i,I-i} \left(\prod_{j=I-i}^{J-1}f_{j}-1\right).
\end{equation}
Thus, going from the CL method to the CC method means that we replace the
observed payment $C_{i,I-i}$ in front of the round bracket in \eqref{CL predictor 00}  by its estimated expected value 
$\beta_{I-i}\,\widehat{\kappa}^{\rm CC} \,\pi_i$, see \eqref{CC predictor 00} and Model Assumptions \ref{Model-Assumptions-Cape-Cod}. Naturally, this estimate $\beta_{I-i}\,\widehat{\kappa}^{\rm CC} \,\pi_i$ is more robust than $C_{i,I-i}$ because it averages over multiple observations.
Note that we have
\begin{equation}\label{CC kappa 3}
\widehat{\kappa}^{\rm CC}
=\sum_{i=1}^I \frac{\beta_{I-i}\, \pi_i}{\sum_{k=1}^I
  \beta_{I-k}\, \pi_k}\, \widehat{\kappa}_i
=\sum_{i=1}^I \frac{\beta_{I-i}\, \pi_i}{\sum_{k=1}^I
\beta_{I-k}\, \pi_k}\, 
\frac{\widehat{C}^{\rm CL}_{i,J}}
{\pi_i},
\end{equation}
thus, the CC claims ratio estimate $\widehat{\kappa}^{\rm CC}$ is the weighted average between the CL predictors divided by the premiums.

Finally, from the above considerations we concluded that Mack's CL model gives an example that complies with
Model Assumptions \ref{Model-Assumptions-Cape-Cod}.
This connection is used for estimating the claims development pattern
$(\beta_j)_{0\leq j\leq J}$ by using Mack's CL factor estimates \cite{Mack} at time $I$ given by
\begin{equation}\label{Mack's CL factor estimates}
\widehat{f}_j = \frac{\sum_{i=1}^{I-j-1} C_{i,j+1}}{\sum_{i=1}^{I-j-1} C_{i,j}}\qquad \text{ for $0 \le j\le  J-1$}.
\end{equation}
These estimates are plugged into \eqref{pattern known CL} for estimating the claims development pattern, giving us the estimates
\begin{equation}\label{pattern unknown CL}
\widehat{\beta}_{j}=  \prod_{k=j}^{J-1}\widehat{f}^{-1}_{k} \qquad \text{ for $0 \le j\le J$}.
\end{equation}
This is the industry standard used, and precisely from this step on we start to deviate from the derivations in Saluz \cite{Saluz}, because she is using an estimate different from \eqref{pattern unknown CL}.


\subsection{Generalized Cape Cod method}
The generalized Cape Cod method (GCC) of Gluck \cite{Gluck} adds an exponential weighting to the CC estimator \eqref{CC kappa 3} for trending, by measuring the time lag to the period that we try to predict. 
This trending is motivated as follows. After \eqref{pattern known CL}, we have considered the identity $\kappa \pi_i = \E[C_{i,0}]/\beta_0$ to connect the CC method to the CL assumption. This consideration requires a stationary claims ratio $\kappa$ and exposures $\pi_i$ that properly discriminate expected ultimate claims $\E[C_{i,J}]$ for different accident years $i$. Since in practice these assumptions are hardly fulfilled, the GCC method of Gluck \cite{Gluck} allows for accident year dependent claims ratio estimates, partly correcting for exposure levels that are not fully correctly specified w.r.t.~the expected ultimate claims differences.
For this, we select a fixed a constant (trending) factor $\lambda \in [0,1]$.
For accident year $1\le i \le I$, we set for the GCC claims ratio estimates
\begin{equation}\label{GCC claims ratio estimate}
\widehat{\kappa}_i^{\rm GCC}(\lambda)
~=~\sum_{l=1}^I \frac{\beta_{I-l}\, \pi_l \, \lambda^{|i-l|}}{\sum_{k=1}^I
\beta_{I-k}\, \pi_k\, \lambda^{|i-k|}}\, \widehat{\kappa}_l
~=~\sum_{l=1}^I \frac{\beta_{I-l}\, \pi_l \, \lambda^{|i-l|}}{\sum_{k=1}^I
\beta_{I-k}\, \pi_k\, \lambda^{|i-k|}}\, \frac{\widehat{C}^{\rm CL}_{l,J}}{\pi_l}.
\end{equation}
The following lemma is again straightforward.
\begin{lemma}For all $\lambda \in [0,1]$ and $1\le i \le I$, we have unbiasedness
  \begin{equation*}
    \E \left[\widehat{\kappa}_i^{\rm GCC}(\lambda) \right] = \kappa.
  \end{equation*}
\end{lemma}

The GCC claims ratios allow for trending using an exponentially decaying rate $\lambda^{|i-l|}$ measuring the distance of the considered accident year $l$ to the selected one $i$. Naturally, if we think of
business cycles which will affect the strength of the premiums $\pi_i$, it makes sense to give more weight to neighboring accounting/accident years. Moreover, if we think of trends in claims due to changes in frequency, severity or portfolio mix, it also makes sense to give more weight to loss experience in closer accident years.

We define the GCC predictor of Gluck \cite{Gluck} as follows.
\begin{defi}[GCC predictor]
Select $\lambda \in [0,1]$. The GCC predictor for the ultimate claim at time $I$ for accident year
$I-J+1 \le i \le I$ is defined by
\begin{equation}\label{GCC predictor first}
\widehat{C}_{i,J}^{\rm GCC}(\lambda)= C_{i,I-i} + \left(1-\beta_{I-i}\right) 
\widehat{\kappa}_i^{\rm GCC}(\lambda) \,\pi_i.
\end{equation}
\end{defi}

We start by giving some immediate properties of the GCC predictor.
Naturally, we have the limiting behaviors for all $1\le i \le I$
\begin{equation*}
\lim_{\lambda \uparrow 1} \widehat{\kappa}_i^{\rm GCC}(\lambda)
=\widehat{\kappa}^{\rm CC}
\qquad \text{ and } \qquad
\lim_{\lambda \downarrow 0} \widehat{\kappa}_i^{\rm GCC}(\lambda)
=\widehat{\kappa}_i.
\end{equation*}
This gives us the following two boundary cases.
\begin{cor} For $I-J+1 \le i \le I$, we have
\begin{equation}\label{CC predictor 0}
\lim_{\lambda \uparrow 1} \widehat{C}_{i,J}^{\rm GCC}(\lambda)=\widehat{C}_{i,J}^{\rm GCC}(1)
=\widehat{C}_{i,J}^{\rm CC},
\end{equation}
and
\begin{equation}
\lim_{\lambda \downarrow 0} \widehat{C}_{i,J}^{\rm GCC}(\lambda)=\widehat{C}_{i,J}^{\rm GCC}(0)
=\label{CL predictor 0}
C_{i,I-i}/\beta_{I-i}
=\widehat{C}_{i,J}^{\rm CL},
\end{equation}
using parameter identification \eqref{pattern known CL}.
\end{cor}
Under Mack's CL Model Assumptions \ref{distribution-free CL model assumptions}, there is the ultimate claim predictor at time $I$
\begin{equation*}
\widehat{C}_{i,J}^{\rm CL}
= \E \left[\left. C_{i,J} \right| C_{i,0}, \ldots, C_{i,I-i}\right]
=C_{i,I-i}\prod_{j=I-i}^{J-1} f_j=C_{i,I-i}/\beta_{I-i}.
\end{equation*}
As a result of \eqref{CC predictor 0} and \eqref{CL predictor 0}, the GCC method continuously interpolates between the CC and the CL predictors as a function of $\lambda \in [0,1]$.

\subsection{Credibility reformulations of the Generalized Cape Cod method}
The previous CL interpretation \eqref{CL predictor 0} allows us to rewrite the GCC ratio estimate as follows
\begin{equation*}
\widehat{\kappa}_i^{\rm GCC}(\lambda)
  =
\sum_{l=1}^I \frac{\beta_{I-l}\, \pi_l \, \lambda^{|i-l|}}{\sum_{k=1}^I
\beta_{I-k}\, \pi_k\, \lambda^{|i-k|}}\, \frac{\widehat{C}_{l,J}^{\rm CL}}{\pi_l}
~=~
\sum_{l=1}^I \alpha_{i,l}(\lambda)\, \frac{\widehat{C}_{l,J}^{\rm CL}}{\pi_l},
\end{equation*}
with positive weights
\begin{equation}\label{weights alpha}
\alpha_{i,l}(\lambda)=\frac{\beta_{I-l}\, \pi_l \, \lambda^{|i-l|}}{\sum_{k=1}^I
\beta_{I-k}\, \pi_k\, \lambda^{|i-k|}} ~\in~ (0,1).
\end{equation}
This allows us to rewrite the GCC predictor as a weighted average
\begin{eqnarray}\nonumber
\widehat{C}_{i,J}^{\rm GCC}(\lambda)&=& C_{i,I-i} + \left(1-\beta_{I-i}\right) 
\widehat{\kappa}_i^{\rm GCC}(\lambda) \,\pi_i
\\&=&\label{second representation of GCC}
\sum_{l=1}^I \alpha_{i,l}(\lambda)\, \left(C_{i,I-i} + \frac{\pi_i}{\pi_l} \left(1-\beta_{I-i}\right) 
\widehat{C}_{l,J}^{\rm CL}\right).
\end{eqnarray}
This has an interesting interpretation. The GCC predictor is a weighted average, $(\alpha_{i,l}(\lambda))_{1\le l  \le I}$, over the CL predictors $\widehat{C}_{l,J}^{\rm CL}$ of all accident years $1\le l \le I$, adjusted to the premium level $\pi_i/\pi_l$ of the considered accident year $i$.

\begin{rems}\label{remarks behavior}\normalfont
The weights $\alpha_{i,l}(\lambda)$ defined in \eqref{weights alpha} consider three different effects. First, they are decreasing in the chronological distance between accident year $l$ and the selected accident year $i$ through the weightings $\lambda^{|i-l|}$ for $\lambda\in (0,1)$. Second, they assign higher credibility to the CL predictors $\widehat{C}_{l,J}^{\rm CL}$
of more mature accident years, assumed that the development pattern $(\beta_j)_{0 \le j \le J}$ is non-decreasing. Third, accident years with a higher exposure, bigger $\pi_l$, receive a higher weight.
\end{rems}

From \eqref{second representation of GCC}, we can even go one step further
by rewriting the first term in the bracket on the right-hand side.
\begin{cor}
We have the following equivalent formulations
for $1 \le i \le I$
\begin{eqnarray*}\nonumber
\widehat{C}_{i,J}^{\rm GCC}(\lambda)&=&
\sum_{l=1}^I \alpha_{i,l}(\lambda)\, \left(\beta_{I-i}\,\widehat{C}_{i,J}^{\rm CL} + \left(1-\beta_{I-i}\right)\frac{\pi_i}{\pi_l}\, 
\widehat{C}_{l,J}^{\rm CL}\right)
\\&=&\nonumber
\omega_{i,i}(\lambda)\,\widehat{C}_{i,J}^{\rm CL}+
\sum_{l=1,\, l\neq i}^I \omega_{i,l}(\lambda)\,\frac{\pi_i}{\pi_l}\, 
\widehat{C}_{l,J}^{\rm CL},
\end{eqnarray*}
with credibility weights
\begin{equation*}
\omega_{i,l}(\lambda) = \beta_{I-i} \,\mathds{1}_{\{i=l\}} + \left(1-\beta_{I-i}\right) \alpha_{i,l}(\lambda)~\in~[0,1].
\end{equation*}
\end{cor}

The focus in this corollary is on expressing the GCC predictor in terms of the CL predictors.
Alternatively, we can express the GCC claims ratio estimates \eqref{GCC claims ratio estimate}
as a credibility weighted average between accident year $i$ and the other accident years' experiences
\begin{equation}\label{another credibility decomposition}
\widehat{\kappa}_i^{\rm GCC}(\lambda)
= \frac{\beta_{I-i}\, \pi_i}{\sum_{k=1}^I
  \beta_{I-k}\, \pi_k\, \lambda^{|i-k|}}\, \widehat{\kappa}_i+
\left(1-\frac{\beta_{I-i}\, \pi_i}{\sum_{k=1}^I
\beta_{I-k}\, \pi_k\, \lambda^{|i-k|}}\right)
  \frac{\sum_{l=1, \, l\neq i}^IC_{l,I-l}\lambda^{|i-l|}}{\sum_{k=1, \, k\neq i}^I\beta_{I-k }\, \pi_k\, \lambda^{|i-k|}}.
\end{equation}
This decomposition again reflects the three-dimensional behavior discussed in Remarks \ref{remarks behavior}. For the boundary case $\lambda \downarrow 0$, it converges to the CL method, and for 
$\lambda \uparrow 1$ to the CC method, the latter can be interpreted as full data pooling.

\section{Prediction uncertainty}
\label{sec: Prediction uncertainty}
\subsection{Mack's CL prediction uncertainty formula}
Before focusing on the CC and the GCC predictors, we recall Mack's CL prediction uncertainty formula; see Mack \cite[Theorem 3]{Mack}.
We come back to the CL predictors  
at time $I$ given by
\begin{equation}\label{CL predictor true}
\widehat{C}_{i,J}^{\rm CL}
= \E \left[\left. C_{i,J} \right| C_{i,0}, \ldots, C_{i,I-i}\right]
= C_{i,I-i}\prod_{j=I-i}^{J-1} f_j.
\end{equation}
The general difficulty in practice is that the CL factors 
$(f_j)_{j=0}^{J-1}$ are unknown and they need to be replaced by their estimates
\eqref{Mack's CL factor estimates} at time $I$. This gives the (estimated) CL predictors at time $I$ (we use a double-hat notation to emphasize the difference to \eqref{CL predictor true}) 
\begin{equation}\label{CL predictor estimated}
\widehat{\widehat{C}_{i,J}^{\rm CL}}
= \widehat{\E} \left[\left. C_{i,J} \right| C_{i,0}, \ldots, C_{i,I-i}\right]
=C_{i,I-i}\prod_{j=I-i}^{J-1} \widehat{f}_j,
\end{equation}
this uses estimates \eqref{Mack's CL factor estimates}. Mack \cite{Mack} derived an estimate of the prediction error in terms of 
an estimated conditional mean squared error of prediction (MSEP) defined as follows
\begin{eqnarray*}
\operatorname{msep}_{\sum_ i C_{i,J}|{\cal D}_I}\left(\sum_i 
\widehat{\widehat{C}_{i,J}^{\rm CL}}\right)
&=& \E \left[\left.\left(\sum_i C_{i,J}
-\sum_i 
\widehat{\widehat{C}_{i,J}^{\rm CL}} \right)^2 \right| {\cal D}_I \right]
\\&=&
\underbrace{\operatorname{Var}\left(\left. \sum_i C_{i,J} \right| {\cal D}_I\right)}_{=: \Psi^2 \text{ (process uncertainty)}}
+
\underbrace{\left(\sum_i \widehat{C}^{\rm CL}_{i,J}
-\sum_i 
\widehat{\widehat{C}_{i,J}^{\rm CL}} \right)^2}_{=: \Delta_{\rm CL}^2 \text{ (parameter estimation error)}},
\end{eqnarray*}
where ${\cal D}_I= \{C_{i,j};\, i+j \le I\}$ denotes the available observations at time $I$. Remark that the process uncertainty term $\Psi^2$ is model-free (at this stage) and the parameter estimation error term $\Delta_{\rm CL}^2$ is model dependent because it explicitly uses the CL predictors $\widehat{\widehat{C}_{i,J}^{\rm CL}}$.

Under Mack's CL Model Assumptions \ref{distribution-free CL model assumptions}, one can explicitly compute the process uncertainty term $\Psi^2$ -- we add a lower index $\Psi^2_{\rm CL}$ because now it becomes model specific -- and replacing all parameters by their estimates
provides us with the estimate (within the CL model), see Mack \cite[Theorem 3]{Mack} and W\"uthrich--Merz \cite[Lemma 3.6]{WM2008},
\begin{equation}\label{process Mack}
\widehat{\Psi}_{\rm CL}^2 = \sum_{i=I-J+1}^I
\left(\widehat{\widehat{C}_{i,J}^{\rm CL}}\right)^2\sum_{j=I-i}^{J-1} 
\frac{\widehat{\sigma}_j^2/\widehat{f}_j^2}{\widehat{\widehat{C}_{i,j}^{\rm CL}}},
\end{equation}
for the variance parameter estimates $(\widehat{\sigma}_j^2)_{j=0}^{J-1}$ given by, we refer to Mack \cite[Section 3]{Mack} and W\"uthrich--Merz \cite[Section 3.2]{WM2008},
\begin{equation*}
\widehat{\sigma}_j^2=
\frac{1}{I-(j+1)}\sum_{i=1}^{I-(j+1)}C_{i,j}\left(\frac{C_{i,j+1}}{C_{i,j}}-
   \widehat{f}_j\right)^2 \qquad \text{ for $j=0,\ldots, J-2$,}
\end{equation*}
and for $j=J-1$, we set $\widehat{\sigma}_{J-1}^2 = \min \{\widehat{\sigma}_{J-2}^4/\widehat{\sigma}_{J-3}^2,\widehat{\sigma}_{J-3}^2,\widehat{\sigma}_{J-2}^2\}$.

\medskip

The parameter estimation error term $\Delta_{\rm CL}^2$ is more difficult, and Mack's \cite{Mack} main achievement was to provide an estimate for this term that is based on a martingale argument. Mack's \cite{Mack} estimate is given by
\begin{equation}\label{parameter Mack}
\widehat{\Delta}_{\rm CL}^2= \sum_{i=I-J+1}^I
\left(\widehat{\widehat{C}_{i,J}^{\rm CL}}\right)^2
\sum_{j=I-i}^{J-1}
\frac{\widehat{\sigma}_j^2/\widehat{f}_j^2}
{\sum_{\ell=1}^{I-j-1}C_{\ell,j}}+2\!\!
\sum_{I-J+1\le i < k \le I}
\widehat{\widehat{C}_{i,J}^{\rm CL}}
\widehat{\widehat{C}_{k,J}^{\rm CL}}
\sum_{j=I-i}^{J-1}
\frac{\widehat{\sigma}_j^2/\widehat{f}_j^2}
{\sum_{\ell=1}^{I-j-1}C_{\ell,j}}.
\end{equation}
Interestingly, R\"ohr \cite{Rohr} derived the identical parameter estimation error estimate
\eqref{parameter Mack} using a rather different method that is based on studying error propagation from physics. The main idea is to consider the impact of parameter estimation error by analyzing a first order Taylor expansion of 
$\Delta_{\rm CL}$. For this, we interpret the CL predictor as a function
of the CL factors $(f_j)_{j=0}^{J-1}$, i.e.,
\begin{equation*}
\widehat{C}_{i,J}^{\rm CL}=
\widehat{C}_{i,J}^{\rm CL}\left((f_j)_{j=0}^{J-1}\right).
\end{equation*}
This gives us estimated CL predictors
\begin{equation*}
\widehat{\widehat{C}_{i,J}^{\rm CL}}=
\widehat{C}_{i,J}^{\rm CL}\left((\widehat{f}_j)_{j=0}^{J-1}\right).
\end{equation*}
Considering a first Taylor expansion to $\Delta_{\rm CL}$ around $(f_j)_{j=0}^{J-1}$ gives us the approximation
\begin{equation}\label{Taylor Rohr CL}
\Delta_{\rm CL} = \sum_i\widehat{C}_{i,J}^{\rm CL}\left((f_j)_{j=0}^{J-1}\right)
-\widehat{C}_{i,J}^{\rm CL}\left((\widehat{f}_j)_{j=0}^{J-1}\right)
~\approx~ - \sum_{j=0}^{J-1} \partial_{f_j}\left[\sum_i
\widehat{C}_{i,J}^{\rm CL}\left((f_j)_{j=0}^{J-1}\right)\right]
\left(\widehat{f}_j - f_j \right).
\end{equation}
This approximation then allowed R\"ohr \cite{Rohr} to estimate 
$\Delta_{\rm CL}^2$ under Mack's CL assumptions which exactly provides
\eqref{parameter Mack}. We use this error propagation method in a similar spirit to R\"ohr \cite{Rohr} in order to assess the estimation error within the GCC prediction framework.


\subsection{Error propagation in the Cape Cod method}
\label{sec: Error propagation in the Cape Cod method}

Using the error propagation method of R\"ohr \cite{Rohr} explained in the previous section, we interpret the GCC predictor \eqref{GCC predictor first} as a function of the CL factors $(f_j)_{j=0}^{J-1}$ using the identification \eqref{pattern known CL}. That is, 
\begin{eqnarray*}
\widehat{C}_{i,J}^{\rm GCC}(\lambda)&=&
                                        \widehat{C}_{i,J}^{\rm GCC}\left((f_j)_{j=0}^{J-1}\right)
                                        \\&=&  C_{i,I-i} + \left(1-\prod_{j=I-i}^{J-1}f^{-1}_{j}\right)\left[ \sum_{l=1}^I \frac{C_{l,I-l} \, \lambda^{|i-l|}}{\sum_{k=1}^I \prod_{j=I-k}^{J-1}f^{-1}_{j}
\, \pi_k\, \lambda^{|i-k|}} \right]\pi_i.
\end{eqnarray*}
Our goal is to understand the sensitivities in the CL factors so that we can analyze the first order Taylor approximation
\begin{eqnarray}\nonumber
\Delta_{\rm GCC}(\lambda) &:=& \sum_{i=1}^I\widehat{C}_{i,J}^{\rm GCC}\left((f_j)_{j=0}^{J-1}\right)
-\widehat{C}_{i,J}^{\rm GCC}\left((\widehat{f}_j)_{j=0}^{J-1}\right)
\\\label{Taylor Rohr GCC}
&\approx& - \sum_{j=0}^{J-1} \partial_{f_j}\left[\sum_{i=1}^I
\widehat{C}_{i,J}^{\rm GCC}\left((f_j)_{j=0}^{J-1}\right)\right]
\left(\widehat{f}_j - f_j \right),
\end{eqnarray}
when replacing the CL factors $f_j$ by their estimates $\widehat{f}_j$.
Our main interest lies in the derivatives of the square bracket on the right-hand side of \eqref{Taylor Rohr GCC}. As common in error propagation, we map this problem to the log-scale. We briefly illustrate the consequences of this.
Consider a differentiable positive function $x\mapsto h(x)>0$. We would like to understand its sensitivity in $x>0$. Naturally, we compute its derivative
\begin{equation*}
  \partial_x h(x) =h'(x).
\end{equation*}
This gives the first order sensitivity in $x>0$. Studying this sensitivity in relative terms requires one to consider the log-scale
\begin{equation*}
  q(x):=\partial_{\log x} \log h(x) = \frac{\partial_{\log x} h(x)}{h(x)}
  = \frac{\partial_{\log x} h(\exp\{\log x\})}{h(x)}
  = \frac{h'(x)x}{h(x)}.
\end{equation*}
This implies 
\begin{equation*}
  \partial_x h(x) =h'(x)=\frac{q(x)}{x}\,h(x).
\end{equation*}
This allows one to represent the first order Taylor approximation around $x$ in relative terms, providing the first order errors
\begin{equation*}
  h(\widehat{x})-h(x)~ \approx ~ h'(x)(\widehat{x}-x) = h(x)\,  \frac{q(x)}{x}\, (\widehat{x}-x).
\end{equation*}
We provide the computation of $q$ in the GCC context to be able to
analyze the right-hand side of \eqref{Taylor Rohr GCC}.
For $t\in \{0, \ldots, J-1\}$, we aim at computing the sensitivities w.r.t.~the CL factor $f_t$
\begin{equation*}
  q_t(\lambda) :=  \partial_{\log f_t} \log \left(\sum_{i=1}^I  \widehat{C}_{i,J}^{\rm GCC}\left((f_j)_{j=0}^{J-1}\right)\right)
  = \frac{\partial_{\log f_t}\left(\sum_{i=1}^I  \widehat{C}_{i,J}^{\rm GCC}\left((f_j)_{j=0}^{J-1}\right)\right)}{\sum_{i=1}^I  \widehat{C}_{i,J}^{\rm GCC}\left((f_j)_{j=0}^{J-1}\right)}.
\end{equation*}

We have the following theorem, its proof is provided in the appendix.

\begin{theo}\label{theorem error propagation}
  Set $I=J+1$. For $\lambda \in [0,1]$ and $t\in \{0, \ldots, J-1\}$, we have
 \begin{equation*}
  q_t(\lambda)   = \frac{\sum_{i=I-t}^I  \beta_{I-i}\,
    \widehat{\kappa}_i^{\rm GCC}(\lambda)\, \pi_i
    + \sum_{i=1}^I \left(1-\beta_{I-i}\right)
    \widehat{\kappa}_i^{\rm GCC}(\lambda)\, \pi_i\,
    \frac{
    \sum_{k=I-t}^I  \beta_{I-k}
    \, \pi_k\, \lambda^{|i-k|}}{\sum_{k=1}^I \beta_{I-k}\, \pi_k\, \lambda^{|i-k|}} }{\sum_{i=1}^I  \widehat{C}_{i,J}^{\rm GCC}(\lambda)}.
\end{equation*} 
\end{theo}

\medskip

Based on Theorem \ref{theorem error propagation}, we obtain the following first order approximation.

\begin{cor}\label{corollary 1st order Taylor}
  For CL factors $(f_j)_{j=0}^{J-1}$ and their perturbed version $(\widehat{f}_j)_{j=0}^{J-1}$, we receive the following first order Taylor approximation for the GCC predictor 
\begin{equation*}
\Delta_{\rm GCC}(\lambda)~\approx~
-\left(\sum_{i=1}^I  \widehat{C}_{i,J}^{\rm GCC}(\lambda)\right)\,
  \sum_{j=0}^{J-1} \frac{q_j(\lambda)}{f_j} \left(\widehat{f}_j - f_j\right).
 \end{equation*}
\end{cor}

\medskip

In the limiting case $\lambda \downarrow 0$ we have
 \begin{equation}\label{q0}
 q_j(0)   =
\lim_{\lambda \downarrow 0}  q_j(\lambda)   = \frac{\sum_{i=1}^I \mathds{1}_{\{j\ge I-i\}}\, \widehat{\kappa}_i \pi_i }{\sum_{i=1}^I  \widehat{C}_{i,J}^{\rm CL}}
= \frac{\sum_{i=I-j}^I \widehat{C}_{i,J}^{\rm CL} }{\sum_{i=1}^I  \widehat{C}_{i,J}^{\rm CL}}.
\end{equation} 
In view of Corollary \ref{corollary 1st order Taylor}, this gives in the CL case
\begin{equation}\label{Taylor Rohr CL 2}
\Delta_{\rm CL} 
~\approx~ - \left(\sum_{i=1}^I
\widehat{C}_{i,J}^{\rm CL}\right)\, \sum_{j=0}^{J-1} \frac{q_j(0)}{f_j}\left(\widehat{f}_j - f_j \right),
\end{equation}
this is precisely the CL formula obtained in W\"uthrich--Merz \cite[page 41]{WM2015}.

\medskip

Applying the same steps as in R\"ohr \cite{Rohr}, we estimate the estimation error term in Gluck's \cite{Gluck} GCC method by
\begin{equation}\label{param error estimate}
\widehat{\Delta}_{\rm GCC}^2(\lambda) :=  \left(\sum_{i=1}^I  \widehat{C}_{i,J}^{\rm GCC}(\lambda)\right)^2\,
\sum_{j=0}^{J-1} \widehat{q}^2_j(\lambda) \,\frac{\widehat{\sigma}^2_j/\widehat{f}_j^2}{\sum_{\ell=1}^{I-j-1}C_{\ell,j}},
\end{equation}
we also refer to the CL estimate \eqref{parameter Mack}, and
where in $\widehat{q}^2_j(\lambda)$ we replace all CL parameters $f_j$ by their estimates
$\widehat{f}_j$, which through \eqref{pattern known CL} is mapped to the claims development pattern.

\begin{cor} We have the following results in the limiting cases
\begin{equation*}
\widehat{\Delta}_{\rm CL}^2=\lim_{\lambda \downarrow 0} 
\widehat{\Delta}_{\rm GCC}^2(\lambda) = 
\left(\sum_{i=1}^I  \widehat{C}_{i,J}^{\rm CL}\right)^2\,
\sum_{j=0}^{J-1} \widehat{q}^2_j(0) \,\frac{\widehat{\sigma}^2_j/\widehat{f}_j^2}{\sum_{\ell=1}^{I-j-1}C_{\ell,j}},
\end{equation*}
this coincides with Mack's \cite{Mack} CL parameter estimation error estimate \eqref{parameter Mack}, and for B\"uhlmann's \cite{Bu_unpublished} CC method we have parameter estimation error estimate
\begin{equation*}
\widehat{\Delta}_{\rm CC}^2=
\lim_{\lambda \uparrow 1} 
\widehat{\Delta}_{\rm GCC}^2(\lambda) = 
 \left(\sum_{i=1}^I  \widehat{C}_{i,J}^{\rm CC}\right)^2\,
\sum_{j=0}^{J-1} \widehat{q}^2_j(1) \,\frac{\widehat{\sigma}^2_j/\widehat{f}_j^2}{\sum_{\ell=1}^{I-j-1}C_{\ell,j}},
\end{equation*}
with $q_j(0)$ given in \eqref{q0} and with
 \begin{equation}\label{q1}
  q_j(1)   
   =\frac{\sum_{k=I-j}^I  \beta_{I-k}
    \, \pi_k }{\sum_{k=1}^I \beta_{I-k}\, \pi_k}~
    \frac{\widehat{\kappa}^{\rm CC}\sum_{i=1}^I
    \pi_i}{\sum_{i=1}^I  \widehat{C}_{i,J}^{\rm CC}}.
\end{equation} 
\end{cor}

The final step is to estimate the process uncertainty. In analogy to \eqref{process Mack}, we set
\begin{equation}\label{process GCC}
\widehat{\Psi}_{\rm GCC}^2(\lambda) = \sum_{i=I-J+1}^I
\left(\widehat{C}_{i,J}^{\rm GCC}(\lambda)\right)^2\sum_{j=I-i}^{J-1} 
\frac{\widehat{\sigma}_j^2/\widehat{f}_j^2}{\widehat{C}_{i,j}^{\rm GCC}(\lambda)}.
\end{equation}
Collecting the terms \eqref{param error estimate} and \eqref{process GCC} provides us with the following estimate for the rooted MSEP (RMSEP) in the GCC method.

\begin{estimator}
The estimate of the RMSEP  in the GCC method for  $\lambda \in [0,1]$ is given by
\begin{equation}\label{GCC RMSEP}
\widehat{\operatorname{msep}}^{1/2}_{\sum_ i C_{i,J}|{\cal D}_I}\left(\sum_i 
  \widehat{C}_{i,J}^{\rm GCC}(\lambda)\right)
= \sqrt{ \widehat{\Psi}_{\rm GCC}^2(\lambda)+\widehat{\Delta}_{\rm GCC}^2(\lambda)},
\end{equation}
with the two terms on the right-hand side defined in
\eqref{process GCC} and \eqref{param error estimate}. 
\end{estimator}
The limiting case $\lambda=0$ in \eqref{GCC RMSEP} gives Mack's \cite{Mack} RMSEP formula in the CL method. However, the limiting case $\lambda=1$ in \eqref{GCC RMSEP} differs from Saluz' \cite{Saluz} RMSEP formula because the estimation of her claims development pattern was not based on the CL method.

\section{Example}
\label{sec: Example}
We present a numerical example and benchmark it with Mack's \cite{Mack} CL prediction and with Saluz' \cite{Saluz} CC method. We use the claims data from W\"uthrich--Merz \cite[Tables 2.2 and 4.3]{WM2008}, the first table provides the observed cumulative payments in the upper triangle $${\cal D}_I= \{C_{i,j};\, i+j \le I, \, 1\le i \le I, \, 0\le j \le J\},$$ with $I=J+1=10$, and the second table gives the premiums $(\pi_i)_{i=1}^I$.

Aggregating over all accident years, we obtain the GCC total ultimate claim predictor at time $I$
\begin{equation*}
\widehat{C}_{\bullet,J}^{\rm GCC}(\lambda)=
\sum_{i=1}^I
\widehat{C}_{i,J}^{\rm GCC}(\lambda) \qquad \text{ for fixed $\lambda \in [0,1]$.}
\end{equation*}
Our focus will be on the outstanding loss liabilities (OLL) at time $I$ given by
\begin{equation*}
R_\bullet= \sum_{i=1}^I C_{i,J}-C_{i,I-i} .
\end{equation*}
These OLL are predicted at time $I$  by the GCC reserves, for fixed $\lambda \in [0,1]$, 
\begin{equation}\label{GCC reserves}
\widehat{R}_\bullet^{\rm GCC}(\lambda)=\widehat{C}_{\bullet,J}^{\rm GCC}(\lambda)-
\sum_{i=1}^I C_{i,I-i}=\sum_{i=1}^I \widehat{C}_{i,J}^{\rm GCC}(\lambda)-C_{i,I-i}.
\end{equation}
We compute these GCC reserves using the estimated CL pattern 
\eqref{Mack's CL factor estimates}-\eqref{pattern unknown CL}. The results are
presented in Table \ref{reserves table} and Figure \ref{fig: GCC reserves}.

\begin{figure}[htb!]
\begin{center}
\begin{minipage}[t]{0.49\textwidth}
\begin{center}
\includegraphics[width=\textwidth]{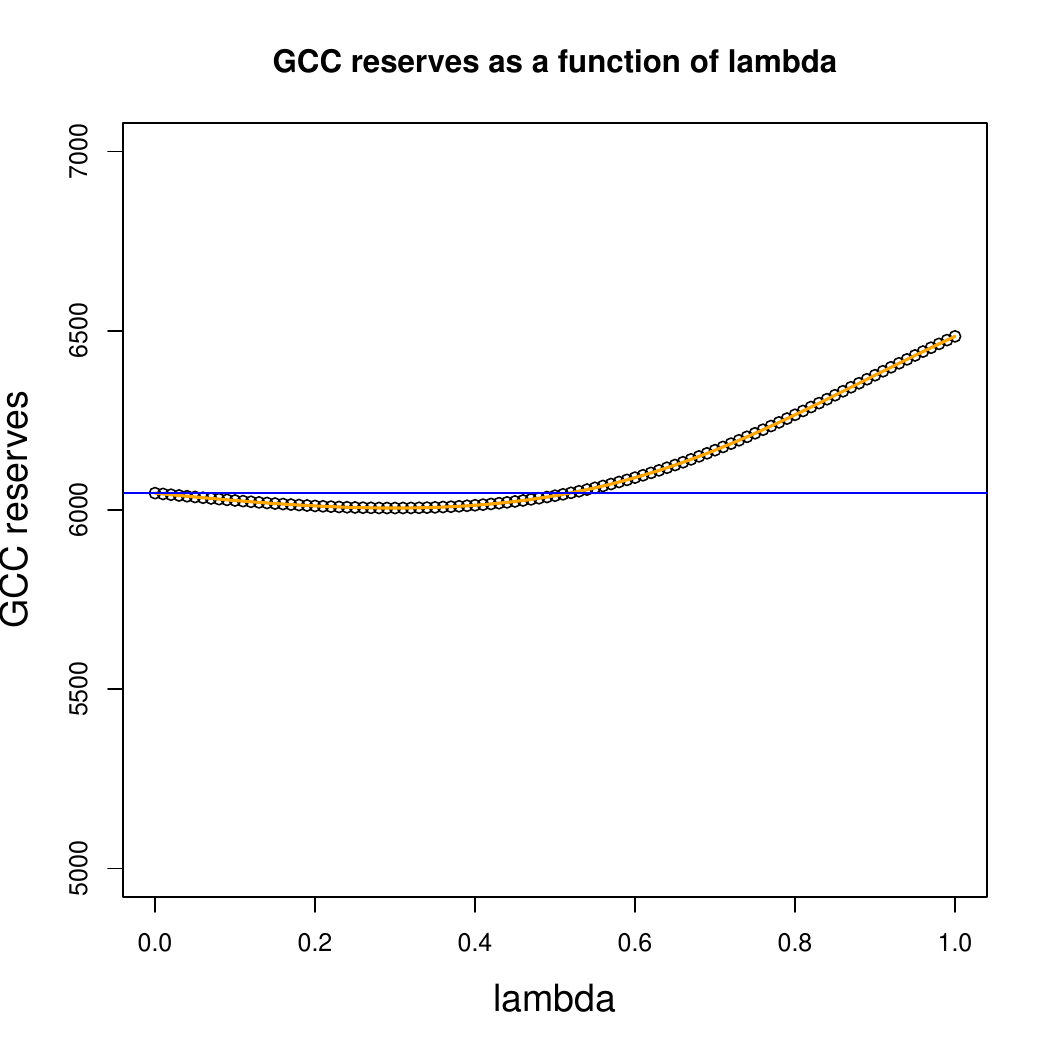}
\end{center}
\end{minipage}
\begin{minipage}[t]{0.49\textwidth}
\begin{center}
\includegraphics[width=\textwidth]{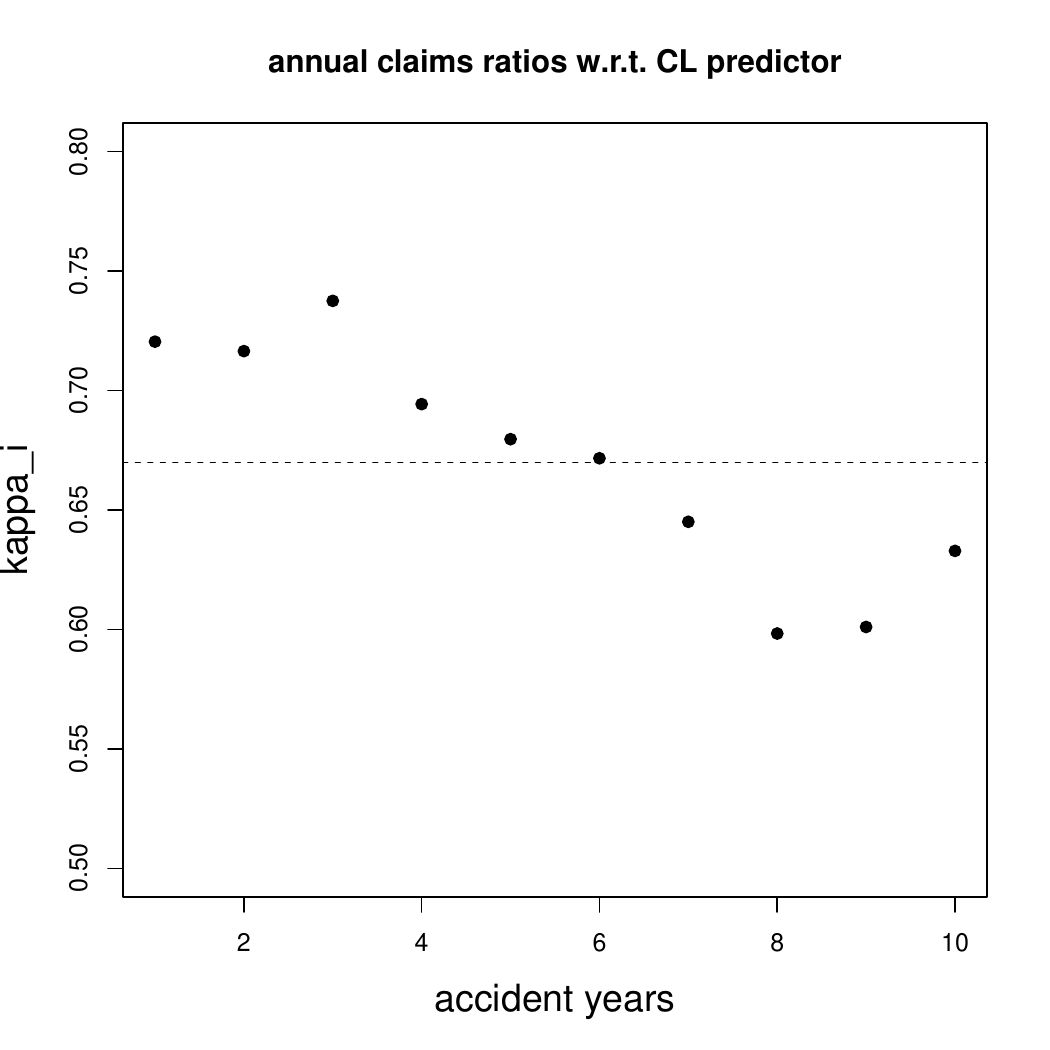}
\end{center}
\end{minipage}
\end{center}
\vspace{-.3cm}
\caption{(lhs) GCC reserves $\widehat{R}_{\bullet}^{\rm GCC}(\lambda)$ as a function of $\lambda \in [0,1]$; $\lambda=0$ gives the CL reserves and $\lambda=1$ the CC reserves; (rhs) individual claims rations $\widehat{\kappa}_i=\widehat{C}^{\rm CL}_{i,J}/\pi_i$ for all accident years $1\le i \le I$. }
\label{fig: GCC reserves}
\end{figure}

\begin{table}[h]
\centering
{\footnotesize
\begin{center}
\begin{tabular}{|l|r|rrr|r|}
  \hline
  method & reserves $ \widehat{R}_\bullet$ & process $\widehat{\Psi}$ & param.~$\widehat{\Delta}$ & RMSEP & CoVa \\ \hline
Mack's \cite{Mack} CL ($\lambda=0$) & 6,047 & 424 & 185 & 463 & 7.66\% \\
GCC with $\lambda=0.25$ & 6,007 & 422 & 174 & 457 & 7.60\% \\
GCC with $\lambda=0.50$ & 6,040 & 422 & 166 & 454 & 7.51\% \\
GCC with $\lambda=0.55$ & 6,062 & 423 & 164 & 454 & 7.48\% \\
GCC with $\lambda=0.75$ & 6,214 & 427 & 159 & 455 & 7.33\% \\
  CC ($\lambda=1$) & 6,485 & 434 & 156 & 461 & 7.11\% \\
  \hline
CC Saluz \cite[Table 2]{Saluz} & 6,618 & 436 & 202 & 481 & 7.26\% \\  
  \hline
\end{tabular}
\end{center}}
\caption{GCC method for selected $\lambda \in [0,1]$, $\lambda=0$ gives the CL and $\lambda=1$ gives the CC methods; the reserves and error estimates are displayed in $10^2$ units.}
\label{reserves table}
\end{table}

The GCC reserves with $\lambda=0$ coincides with Mack's \cite{Mack} CL reserves, see also W\"uthrich--Merz \cite[Table 3.6]{WM2008}, the GCC reserves with $\lambda=1$ coincide with B\"uhlmann's \cite{Bu_unpublished} CC reserves, see also W\"uthrich--Merz \cite[Table 4.3]{WM2008}, but they differ from Saluz' \cite[Table 2]{Saluz} CC reserves, because the latter uses a different estimate for the claims development pattern $(\beta_j)_{j=0}^J$, see Saluz \cite[Formula (2.3)]{Saluz}.

We interpret the results by first focusing on Figure \ref{fig: GCC reserves}.
The right-hand side of Figure \ref{fig: GCC reserves} shows the individual estimated claims ratios
$$
\widehat{\kappa}_i=\frac{C_{i,I-i}}{\widehat{\beta}_{I-i}\,\pi_i}=\frac{\widehat{C}^{\rm CL}_{i,J}}{\pi_i} \qquad \text{ for $i=1,\ldots, I$.}
$$
From this figure we conclude that these individual estimated claims ratios 
$\widehat{\kappa}_i$ are decreasing in $i=1,\ldots, I=10$. Thus, this data has a time trend in calendar year $i$ which is not properly captured in the premiums $\pi_i$ to comply with Model Assumptions \ref{Model-Assumptions-Cape-Cod}. Plugging this into the CC method \eqref{CC kappa} of B\"uhlmann \cite{Bu_unpublished} we likely over-estimate the OLL $R_\bullet$ by the CC reserves $\widehat{R}_\bullet^{\rm CC}:=\widehat{R}_\bullet^{\rm GCC}(1)$. This is also suggested by the left-hand side of Figure \ref{fig: GCC reserves} which plots the GCC reserves $\lambda \mapsto \widehat{R}_\bullet^{\rm GCC}(\lambda)$ as a function of $\lambda \in [0,1]$. We have a comparably stable picture for $\lambda \in [0,0.6]$, which suggest to choosing $\lambda$ within this interval to capture the trend in $\widehat{\kappa}_i$, i.e., this is precisely the additional trending feature offered by Gluck's \cite{Gluck} GCC method. In the range of $\lambda \in [0,0.6]$, the GCC reserves show little sensitivity in the specific $\lambda$-choice, i.e., we have robustness in this hyper-parameter choice which is typically is a good property in practical applications.

\begin{figure}[htb!]
\begin{center}
\begin{minipage}[t]{0.49\textwidth}
\begin{center}
\includegraphics[width=\textwidth]{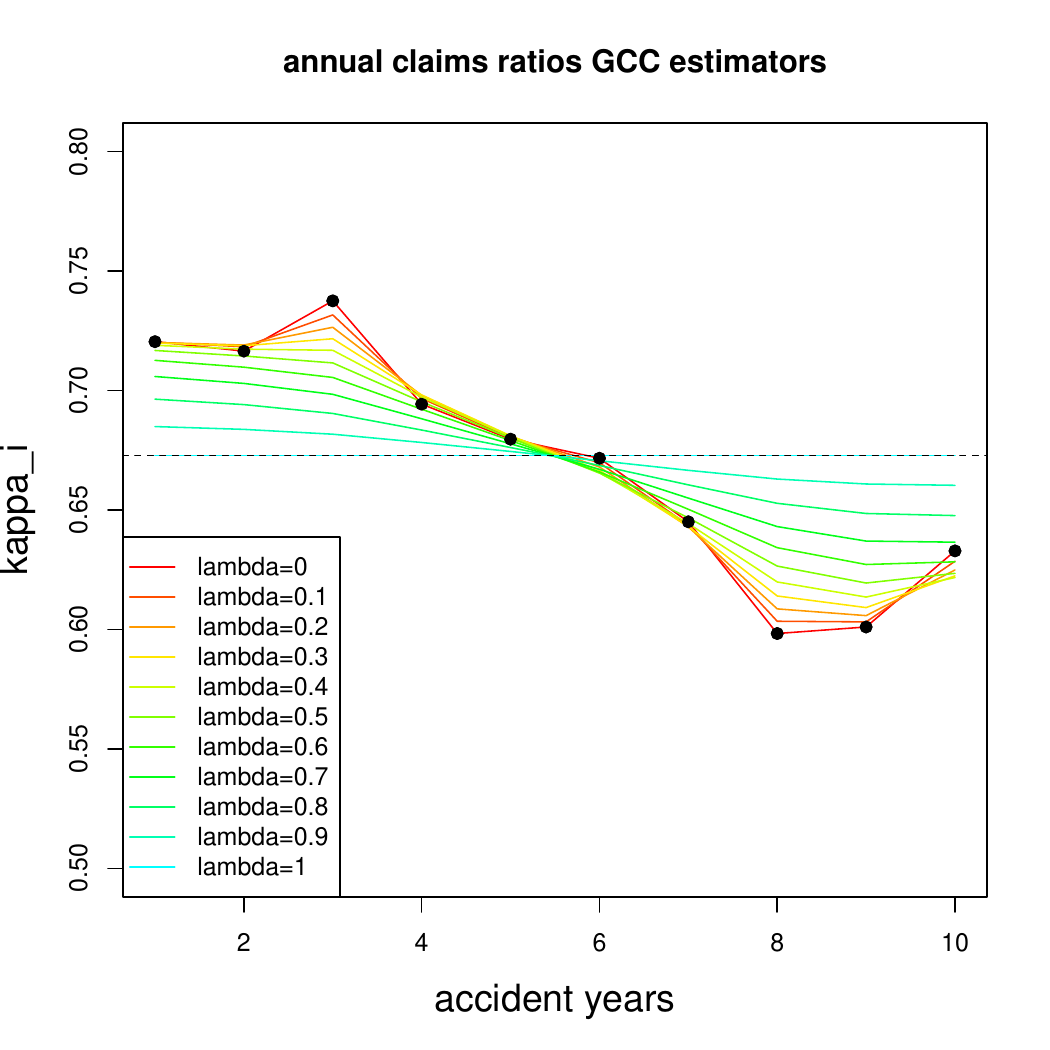}
\end{center}
\end{minipage}
\end{center}
\vspace{-.3cm}
\caption{Estimated GCC claims rations $\widehat{\kappa}^{\rm GCC}_i(\lambda)$ for all accident years $1\le i \le I$ and $\lambda \in \{0,0.1,0.2, \ldots, 1\}$. }
\label{fig: GCC reserves 2}
\end{figure}

Figure \ref{fig: GCC reserves 2} shows the resulting GCC claims ratios 
$\widehat{\kappa}^{\rm GCC}_i(\lambda)$ for $\lambda \in \{0,0.1,0.2, \ldots, 1\}$.
The case $\lambda=0$ gives the CL reserves which corresponds to the black dots in
Figure \ref{fig: GCC reserves 2}, this is fully observation-driven. The case $\lambda=1$ gives the CC reserves corresponding to the horizontal dashed line. The cases $\lambda \in (0,1)$ provide a middle ground being a weighted average of these two extreme cases.

From Table \ref{reserves table}, we also observe that the CC prediction of Saluz \cite{Saluz} probably even more over-estimates than the CC prediction of B\"uhlmann \cite{Bu_unpublished}. This comes from the fact that Saluz \cite[Section 2.1.1]{Saluz} also uses the premiums $(\pi_i)_{i=1}^I$ for the estimation of the claims development pattern $(\beta_j)_{j=0}^J$. Generally, we do not recommend this, since in case of trends, the CC method of Saluz \cite{Saluz} seems to apply a double-penalization of this trend misspecification as can be seen from the results of Table \ref{reserves table}. Using the method of Saluz' \cite{Saluz}, one could develop a method for removing this trend from the premiums, e.g., by constructing an index from the GCC annual claims ratios, but we do not consider this further here.

\begin{figure}[htb!]
\begin{center}
\begin{minipage}[t]{0.49\textwidth}
\begin{center}
\includegraphics[width=\textwidth]{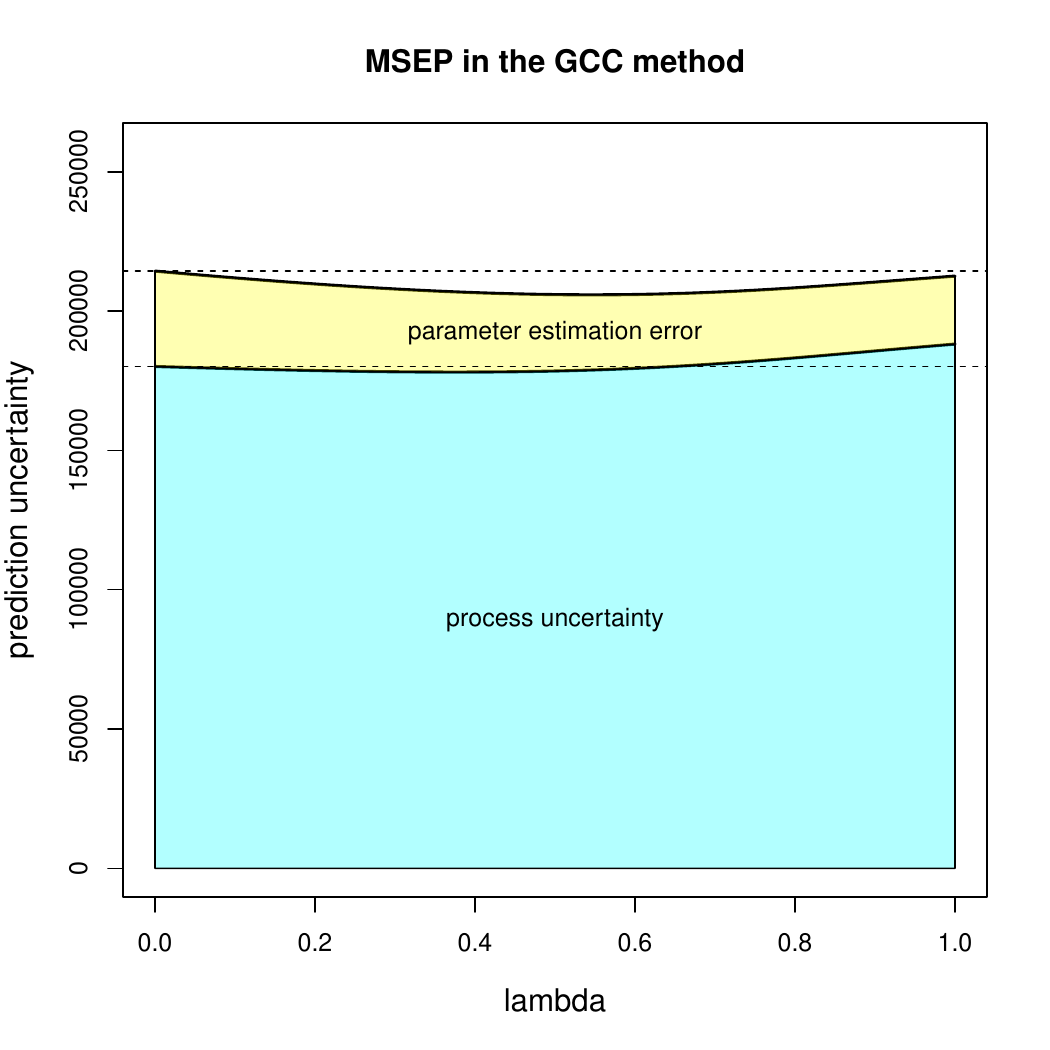}
\end{center}
\end{minipage}
\begin{minipage}[t]{0.49\textwidth}
\begin{center}
\includegraphics[width=\textwidth]{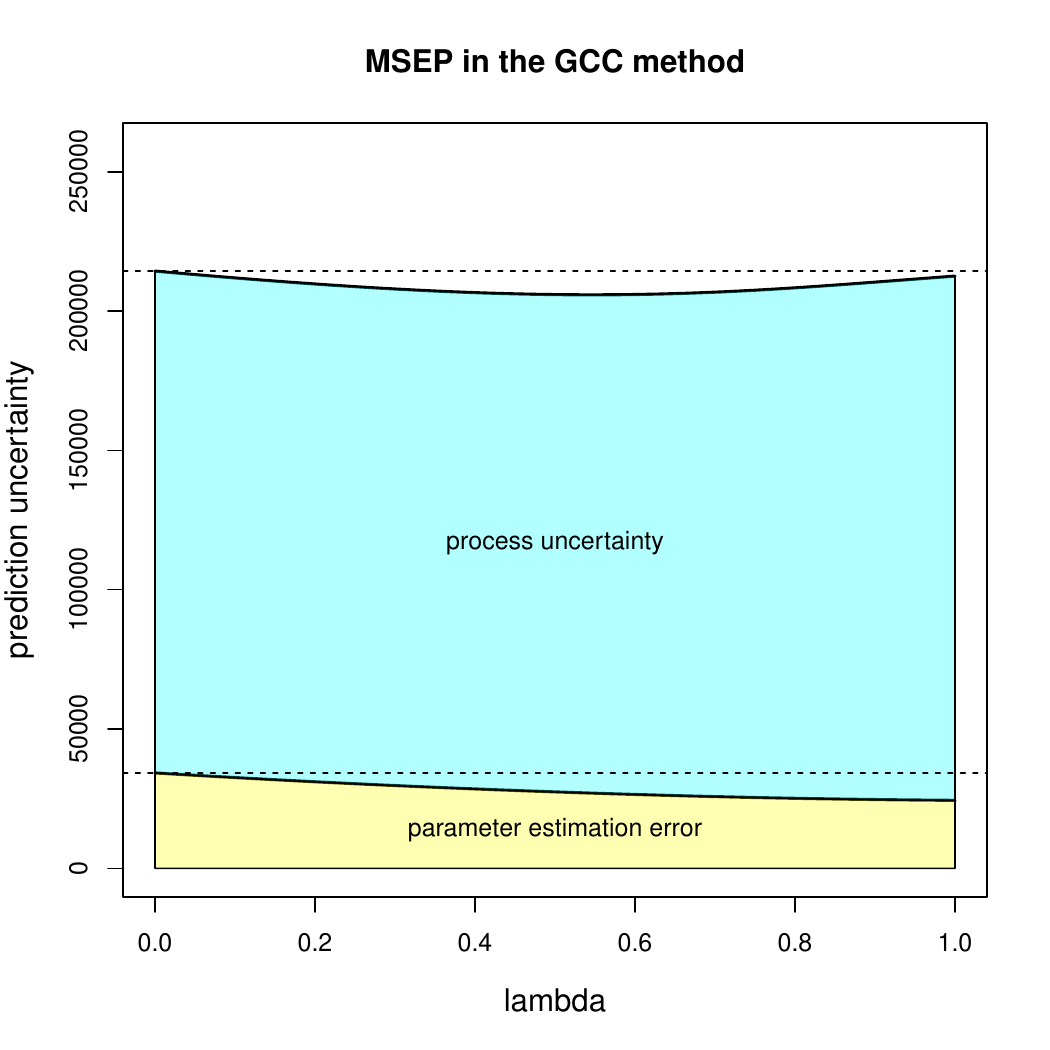}
\end{center}
\end{minipage}
\end{center}
\vspace{-.3cm}
\caption{MSEP $\lambda \mapsto \widehat{\operatorname{msep}}_{C_{\bullet,J}|{\cal D}_I}( 
  \widehat{C}_{\bullet,J}^{\rm GCC}(\lambda))$ as a function of $\lambda \in [0,1]$ and split w.r.t.~process uncertainty $\widehat{\Psi}_{\rm GCC}^2(\lambda)$ and parameter estimation error $\widehat{\Delta}_{\rm GCC}^2(\lambda)$.}
\label{fig: MSEP 1}
\end{figure}

Next, we focus on the MSEP results \eqref{GCC RMSEP} considered as a function of $\lambda \in [0,1]$, with $\lambda=0$ giving us Mack's \cite{Mack} CL MSEP. Figure \ref{fig: MSEP 1} shows these MSEPs (unrooted versions) being split w.r.t.~the process uncertainty $\widehat{\Psi}_{\rm GCC}^2(\lambda)$ and the parameter estimation error $\widehat{\Delta}_{\rm GCC}^2(\lambda)$. The two plots only differ in the order of the arrangement of  
$\widehat{\Psi}_{\rm GCC}^2(\lambda)$ and $\widehat{\Delta}_{\rm GCC}^2(\lambda)$. In this example, we observe that the process uncertainty dominates the parameter estimation error -- there are also many examples where this is not the case. The parameter estimation error $\widehat{\Delta}^2_{\rm GCC}(\lambda)$ is decreasing in $\lambda$, while the process uncertainty $\widehat{\Psi}^2_{\rm GCC}(\lambda)$ is slightly increasing in that parameter. The decreasing estimation error is explained by the fact that with increasing $\lambda$ we increasingly 'trust' more into a correct premium level $\pi_i$ which (in a credibility philosophy) should lower the model error, and hence the parameter estimation error. The increasing process uncertainty is mainly attributed to the fact that the underlying reserves $\widehat{R}_\bullet^{\rm GCC}(\lambda)$ are increasing in $\lambda$. In this example, the minimal MSEP is obtained in $\lambda = 0.55$. Remark that this is simply an observation for which $\lambda$ the minimal MSEP is attained. However, there is no theory that verifies that the model with the minimal MSEP is the optimal one. In fact, a small MSEP does not express model accuracy but often it is just an artefact of the model assumptions and the estimation procedure chosen. Proper model selection can be done, e.g., be the re-reserving approach presented in Balona--Richman \cite{BalonaRichman}. This re-reserving approach presents a proper out-of-sample and out-of-time verification.

\begin{figure}[htb!]
\begin{center}
\begin{minipage}[t]{0.49\textwidth}
\begin{center}
\includegraphics[width=\textwidth]{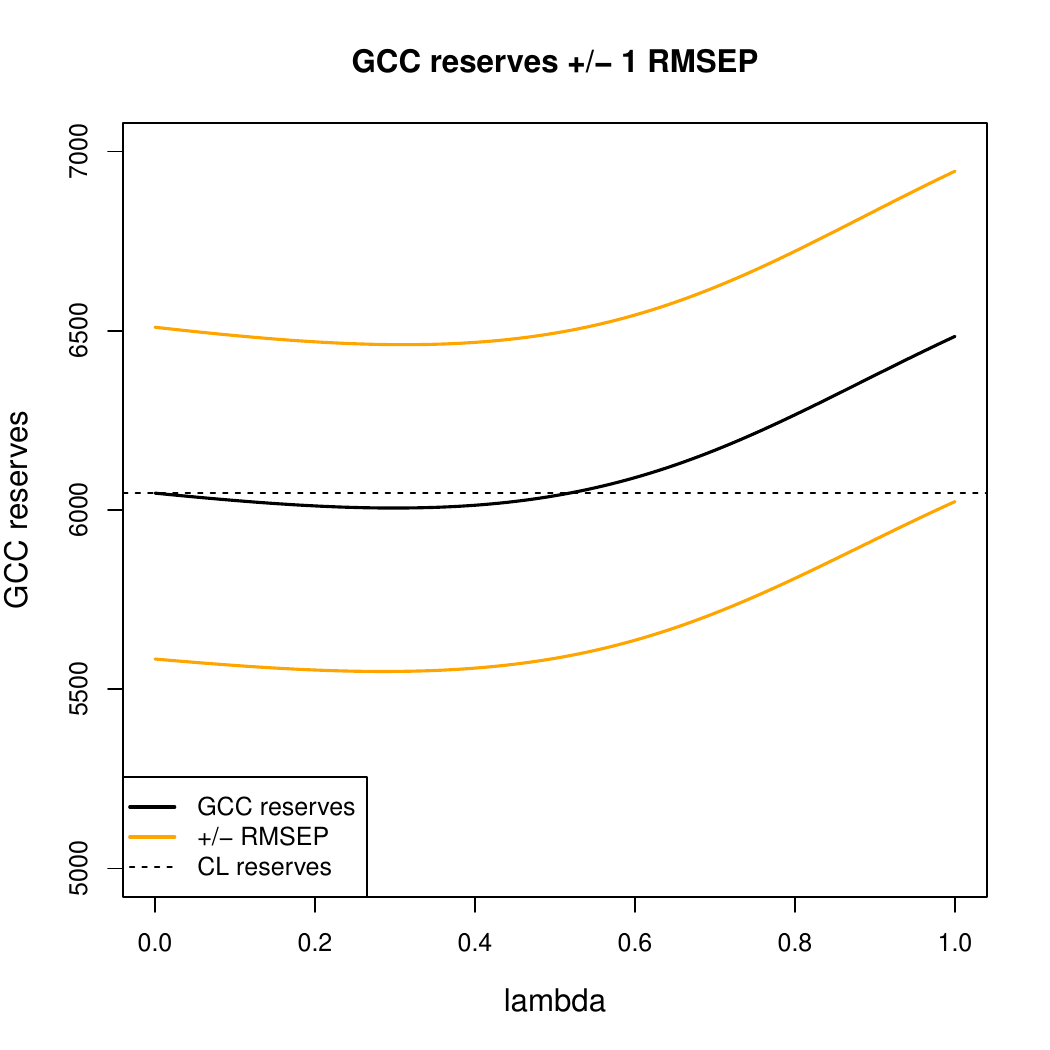}
\end{center}
\end{minipage}
\begin{minipage}[t]{0.49\textwidth}
\begin{center}
\includegraphics[width=\textwidth]{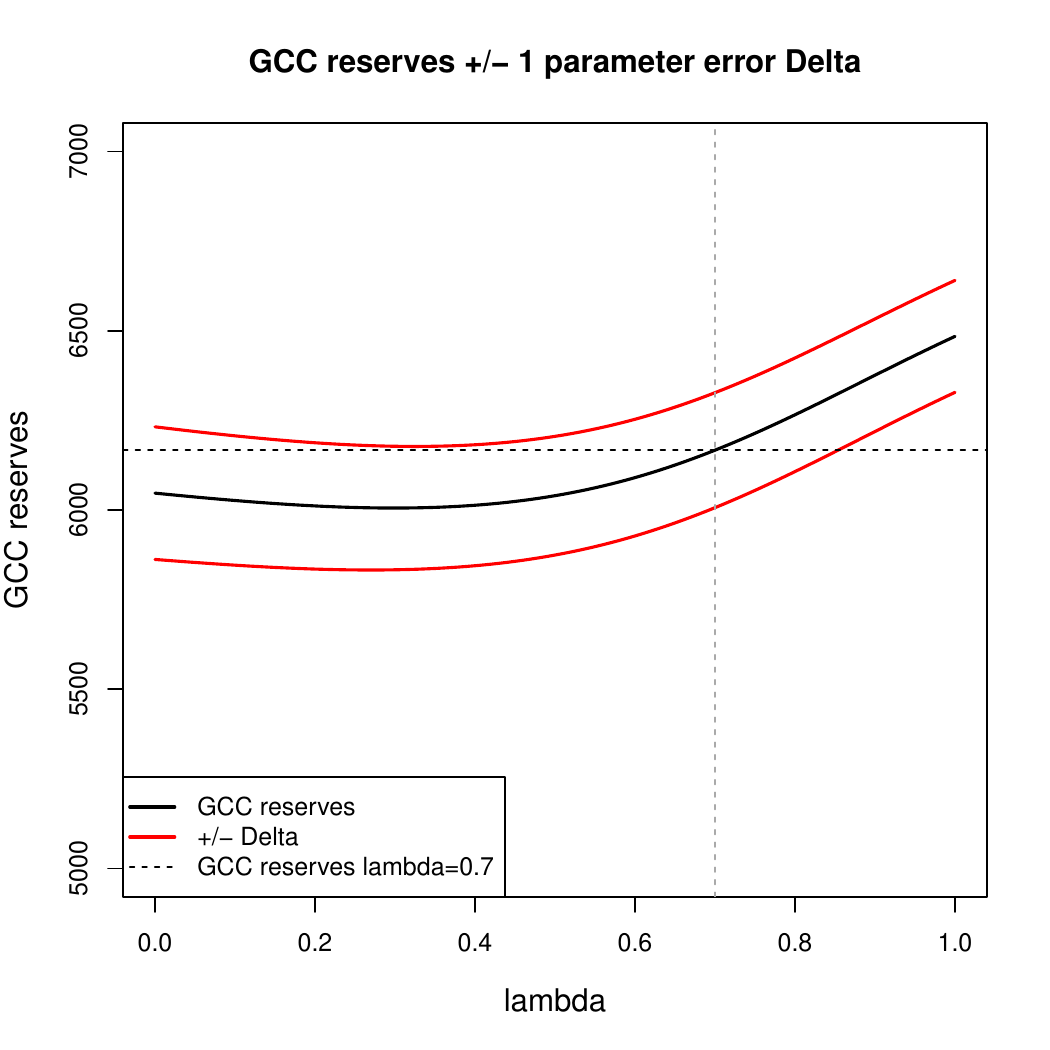}
\end{center}
\end{minipage}
\end{center}
\vspace{-.7cm}
\caption{(lhs) GCC reserves $\lambda \mapsto \widehat{R}_\bullet^{\rm GCC}(\lambda)$ with confidence bounds of one RMSEP $\widehat{\operatorname{msep}}^{1/2}_{C_{\bullet,J}|{\cal D}_I}( 
  \widehat{C}_{\bullet,J}^{\rm GCC}(\lambda))$; the (rhs) gives the same plot as on the (lhs) but the confidence bounds correspond to the estimation error $\widehat{\Delta}_{\rm GCC}(\lambda)$; the $y$-scale is identical in both plots.}
\label{fig: MSEP 2}
\end{figure}

Next, we present the GCC reserves $\widehat{R}_\bullet^{\rm GCC}(\lambda)$ together with their RMSEPs $\widehat{\operatorname{msep}}^{1/2}_{C_{\bullet,J}|{\cal D}_I}(\widehat{C}_{\bullet,J}^{\rm GCC}(\lambda))$ as a function of $\lambda \in [0,1]$. The results are given in Table \ref{reserves table} and Figure \ref{fig: MSEP 2} (left-hand side). We obtain a comparably stable picture for $\lambda \in [0,0.6]$, and for bigger $\lambda$ the trend in premiums leads to increasing claims reserves. In fact, the difference between the CL reserves and the CC reserves is roughly one RMSEP as can be seen from Figure \ref{fig: MSEP 2} (left-hand side). 
To further sharpen the picture, we give the same graph on the right-hand side of Figure \ref{fig: MSEP 2}, but we replace the RMSEP confidence bounds by confidence bounds corresponding to the parameter estimation error $\widehat{\Delta}_{\rm GCC}(\lambda)$; we interpret this quantity as the impact of model uncertainty coming from model fitting (for a given  $\lambda$). The dashed lines show the GCC reserves $\widehat{R}_\bullet^{\rm GCC}(\lambda)$ for $\lambda=0.7$. This is the maximal claims reserves that is within the error bounds for all smaller $\lambda <0.7$ which could still be explained by parameter estimation error. This indicates that the GCC reserves for bigger $\lambda>0.7$ likely over-estimate the OLL, they are rather sensitive to the choice of $\lambda$ in this range.

\begin{figure}[htb!]
\begin{center}
\begin{minipage}[t]{0.49\textwidth}
\begin{center}
\includegraphics[width=\textwidth]{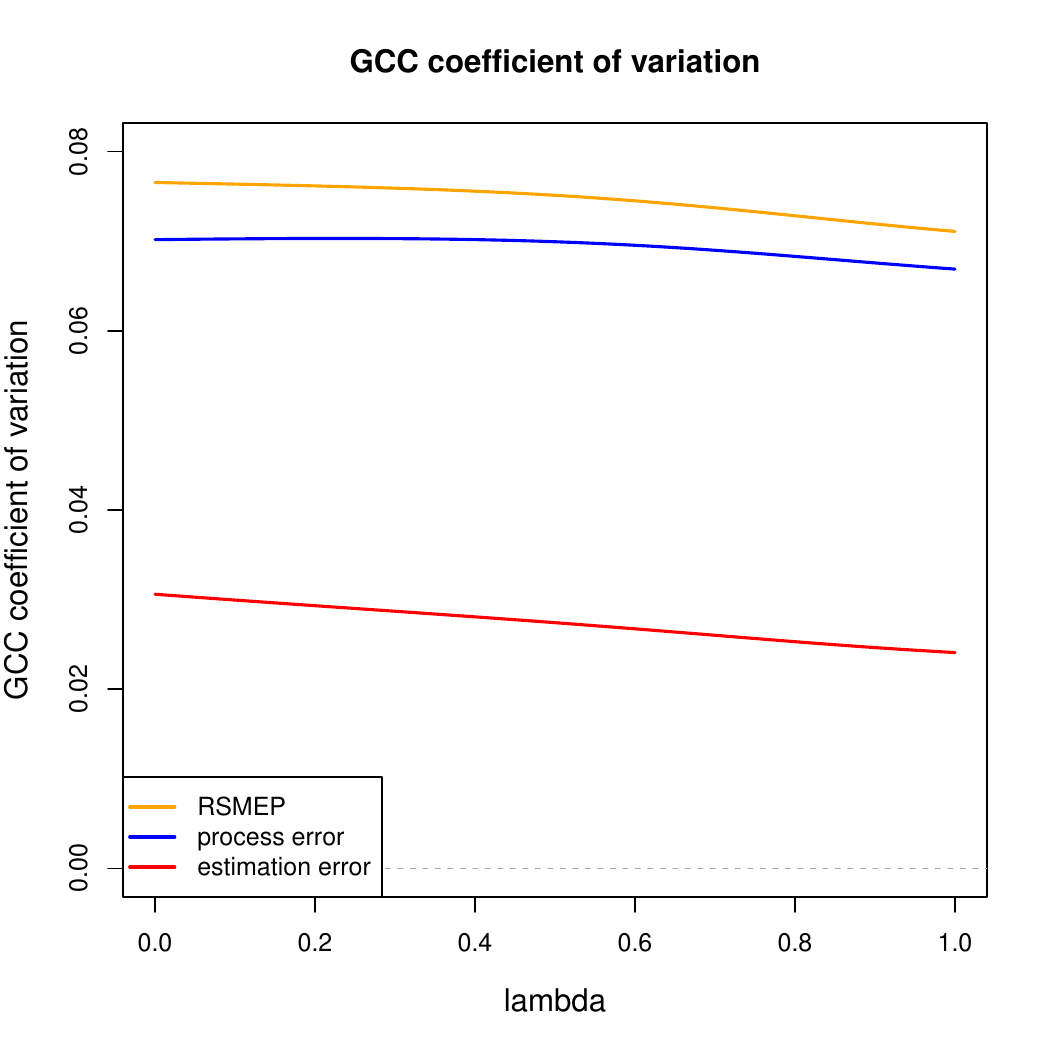}
\end{center}
\end{minipage}
\end{center}
\vspace{-.7cm}
\caption{Coefficient of variation $\operatorname{CoVa}(\lambda)$, additionally we indicate the relative process uncertainty $\widehat{\Psi}_{\rm GCC}(\lambda)/\widehat{R}_\bullet^{\rm GCC}(\lambda)$ and the relative estimation error
  $\widehat{\Delta}_{\rm GCC}(\lambda)/\widehat{R}_\bullet^{\rm GCC}(\lambda)$.}
\label{fig: MSEP 3}
\end{figure}

Finally, Figure \ref{fig: MSEP 3} shows the coefficient of variation (CoVa) defined by
\begin{equation}\label{def: CoVa}
  \operatorname{CoVa}(\lambda) = \frac{\widehat{\operatorname{msep}}^{1/2}_{ C_{\bullet,J}|{\cal D}_I}\left(   \widehat{C}_{\bullet,J}^{\rm GCC}(\lambda)\right)}{\widehat{R}_\bullet^{\rm GCC}(\lambda)} \qquad \text{ for $\lambda \in [0,1]$.}
\end{equation}
Moreover, we also provide this ratio for the process uncertainty and the parameter estimation error only, respectively. From these graphs we conclude that the CoVa is roughly 7.5\%, and likely the CC method under-estimates the relative uncertainty, mainly because the resulting CC reserves (denominator in \eqref{def: CoVa}) are too large.

\section{Conclusions}
\label{sec: Conclusions}

This paper discussed B\"uhlmann's \cite{Bu_unpublished} Cape Cod method and Gluck's \cite{Gluck} generalized Cape Cod extension that allows for trending. We illustrated how these two claims reserving methods are related to Mack's \cite{Mack} chain-ladder claims reserving model.  On the technical side, we provided a new analytical formula for the prediction uncertainty in the generalized Cape Cod method. In fact, this is the literature gap that has now been closed by the present paper. The derivation of this prediction uncertainty formula is based on R\"ohr's \cite{Rohr} error propagation approach. In the limiting case, the generalized Cape Cod approach gives the chain-ladder reserves, and we verified that our prediction uncertainty formula converges to Mack's \cite{Mack} uncertainty formula in the chain-ladder model. On the other hand, our prediction uncertainty formula in the Cape Cod method differs from the one derived in Saluz \cite{Saluz} because this latter reference uses a different parameter estimation method compared to B\"uhlmann's \cite{Bu_unpublished} Cape Cod proposal, which naturally also leads to different prediction uncertainties.

The important practical implication of the new mean squared error of prediction (MSEP) estimator presented here is methodological consistency between booked reserves and the reported uncertainty. Industry reserving processes often book reserves via the Bornhuetter--Ferguson \cite{BF} or the generalized Cape Cod \cite{Gluck} reserving methods or some intermediate position, but report prediction uncertainty using either Mack's \cite{Mack} chain-ladder formula or the Mack bootstrap of England--Verrall \cite{EnglandVerrall2006}. Both of these approaches reproduce the chain-ladder reserves in the mean and, therefore, quantify the uncertainty of a predictor that is typically not the one being booked. The framework presented here provides a closed-form MSEP estimator that uses the same chain-ladder-implied parameter estimates as the generalized Cape Cod predictor, with the generalized Cape Cod method \cite{Gluck} being a good candidate for the booked reserves themselves or serving as a good algorithmic approximation to Bornhuetter--Ferguson \cite{BF} reserving.

\bigskip

{\small 
\renewcommand{\baselinestretch}{.51}
}

\newpage
\appendix

\section{Proofs}

{\Beweis
  {\bf Proof of Theorem \ref{theorem error propagation}.}
The GCC predictor over all accident years is given by
\begin{eqnarray}
  \sum_{i=1}^I  \widehat{C}_{i,J}^{\rm GCC}\left((f_j)_{j=0}^{J-1}\right)
\label{compute derivative 1}
  &=&\sum_{i=1}^I
  C_{i,I-i} + \left(1-\prod_{j=I-i}^{J-1}f^{-1}_{j}\right)\left[ \sum_{l=1}^I \frac{C_{l,I-l} \, \lambda^{|i-l|}}{\sum_{k=1}^I \prod_{j=I-k}^{J-1}f^{-1}_{j}
\, \pi_k\, \lambda^{|i-k|}} \right]\pi_i.
\end{eqnarray}
The right-hand side expresses the GCC predictor as a function of the CL factors $(f_j)_{j=0}^{J-1}$.
We consider the derivative w.r.t.~$\log f_t$, $t\in \{0,\ldots, J-1\}$, 
\begin{equation*}
\partial_{\log f_t}\left(\prod_{j=I-i}^{J-1}f^{-1}_{j}\right)
  = \partial_{\log f_t}\left(\prod_{j=I-i}^{J-1}e^{-\log f_{j}}\right)     
  =- \mathds{1}_{\{t\ge I-i\}}\left(\prod_{j=I-i}^{J-1}f^{-1}_{j}\right).
\end{equation*}
Applying this to \eqref{compute derivative 1}, we obtain for $i \in \{1,\ldots, I\}$
and $t\in \{0,\ldots, J-1\}$
\begin{eqnarray*}
&& \hspace{-1cm}
 \partial_{\log f_t} \left(1-\prod_{j=I-i}^{J-1}f^{-1}_{j}\right)\left[ \sum_{l=1}^I \frac{C_{l,I-l} \, \lambda^{|i-l|}}{\sum_{k=1}^I \prod_{j=I-k}^{J-1}f^{-1}_{j}
    \, \pi_k\, \lambda^{|i-k|}} \right]\pi_i
\\      
&=&  \mathds{1}_{\{t\ge I-i\}} \prod_{j=I-i}^{J-1}f^{-1}_{j}\left[ \sum_{l=1}^I \frac{C_{l,I-l} \, \lambda^{|i-l|}}{\sum_{k=1}^I \prod_{j=I-k}^{J-1}f^{-1}_{j}
    \, \pi_k\, \lambda^{|i-k|}} \right]\pi_i
\\&&+    \left(1-\prod_{j=I-i}^{J-1}f^{-1}_{j}\right)\partial_{\log f_t}\left[  \frac{\sum_{l=1}^I C_{l,I-l} \, \lambda^{|i-l|}}{\sum_{k=1}^I \prod_{j=I-k}^{J-1}f^{-1}_{j}
     \, \pi_k\, \lambda^{|i-k|}} \right]\pi_i.
\end{eqnarray*}     
There remains the computation of the last derivative. It is given by
\begin{eqnarray*}
&& \hspace{-1cm}
\partial_{\log f_t}\left[  \frac{\sum_{l=1}^I C_{l,I-l} \, \lambda^{|i-l|}}{\sum_{k=1}^I \prod_{j=I-k}^{J-1}f^{-1}_{j}
   \, \pi_k\, \lambda^{|i-k|}} \right]
  \\&=&
-\left[  \frac{\sum_{l=1}^I C_{l,I-l} \, \lambda^{|i-l|}}{\left(\sum_{k=1}^I \prod_{j=I-k}^{J-1}f^{-1}_{j}
        \, \pi_k\, \lambda^{|i-k|}\right)^2} \right]\partial_{\log f_t}
        \left(\sum_{k=1}^I \prod_{j=I-k}^{J-1}f^{-1}_{j}
        \, \pi_k\, \lambda^{|i-k|}\right)
  \\&=&\left[  \frac{\sum_{l=1}^I C_{l,I-l} \, \lambda^{|i-l|}}{\sum_{k=1}^I \prod_{j=I-k}^{J-1}f^{-1}_{j}
        \, \pi_k\, \lambda^{|i-k|}} \right]\frac{
        \sum_{k=1}^I \mathds{1}_{\{t\ge I-k\}} \prod_{j=I-k}^{J-1}f^{-1}_{j}
   \, \pi_k\, \lambda^{|i-k|}}{\sum_{k=1}^I \prod_{j=I-k}^{J-1}f^{-1}_{j}
        \, \pi_k\, \lambda^{|i-k|}}.        
\end{eqnarray*}
Collecting all the term provides us with
\begin{eqnarray*}
&& \hspace{-1cm}
 \partial_{\log f_t} \left(1-\beta_{I-i}\right)\left[ \sum_{l=1}^I \frac{C_{l,I-l} \, \lambda^{|i-l|}}{\sum_{k=1}^I \beta_{I-k} \, \pi_k\, \lambda^{|i-k|}} \right]\pi_i
\\      
&=&  \mathds{1}_{\{t\ge I-i\}}\, \beta_{I-i}\,
    \widehat{\kappa}_i^{\rm GCC}(\lambda)\, \pi_i
    + \left(1-\beta_{I-i}\right)
    \widehat{\kappa}_i^{\rm GCC}(\lambda)\, \pi_i\,
    \frac{
    \sum_{k=1}^I \mathds{1}_{\{t\ge I-k\}} \beta_{I-k}
    \, \pi_k\, \lambda^{|i-k|}}{\sum_{k=1}^I \beta_{I-k}\, \pi_k\, \lambda^{|i-k|}}.
\end{eqnarray*}     
This completes the proof.
\EndProof}

\end{document}